\documentclass[12pt]{article}
\usepackage{amsmath}
\usepackage{graphicx}
\usepackage{enumerate}
\usepackage{natbib}
\usepackage{url} 

\usepackage{amssymb}
\usepackage{soul}
\usepackage{xcolor}
\usepackage{algorithm}
\usepackage{algorithmic}
\usepackage{enumitem}
\usepackage{multicol}
\usepackage{amsthm}
\usepackage{bm}

\newcommand{\ex}{{\mathbb{E}\,}}
\newcommand{\var}{{\textrm{Var}\,}}

\newcommand{\vm}{\mathbf}
\newcommand{\vw}{\vm{w}}
\newcommand{\vx}{\vm{x}}

\newcommand{\vM}{\vm{M}}
\newcommand{\vX}{\vm{X}}

\newcommand{\vZ}{\vm{Z}}
\newcommand{\vmu}{\bm{\mu}}
\newcommand{\vnu}{\bm{\nu}}
\newcommand{\vSigma}{\bm{\Sigma}}
\newcommand{\vXi}{\bm{\Xi}}
\newcommand{\sR}{\mathcal{R}}
\newcommand{\hmu}{\hat{\vmu}}
\newcommand{\cmu}{\check{\vmu}}

\newcommand{\vone}{\bm{1}}

\newtheorem{theorem}{Theorem}

\newtheorem{definition}{Definition}

\newtheorem{lemma}{Lemma}
\newtheorem{example}{Example}

\newcommand{\blind}{1}

\begin{document}

\def\spacingset#1{\renewcommand{\baselinestretch}%
{#1}\small\normalsize} \spacingset{1}


\if1\blind
{
  \title{\bf A Statistical Model with Qualitative Input}
  \author{Seksan Kiatsupaibul\\
  		Department of Statistics, \\Chulalongkorn University,\\Bangkok 10330, Thailand \\seksan@cbs.chula.ac.th
  		\and
  		Pariyakorn Maneekul\\
  		Department of Industrial and Systems Engineering, \\University of Washington, \\Seattle, Washington 98195 \\parim@uw.edu
  		}
  \maketitle
} \fi

\if0\blind
{
  \bigskip
  \bigskip
  \bigskip
  \begin{center}
    {\LARGE\bf A Statistical Model with Qualitative Input and an Application to Portfolio Selection}
\end{center}
  \medskip
} \fi

\bigskip
\begin{abstract}
A statistical estimation model with qualitative input provides a mechanism to fuse human 
intuition in the form of qualitative information into a statistical model.  
We investigate the statistical properties of this model and devise a numerical computation method 
for a model subclass with a uniform correlation structure. 
We show that, within this subclass, qualitative information can be as useful as quantitative information.  
We also show that the correlation between variables compromises the accuracy of the statistical estimate.  
However, the adverse effect from the correlation can be minimal, as is illustrated in 
an application to portfolio selection.   
The proposed model, when used in conjunction with approximation techniques, is shown to have potential 
for portfolio selection with financial data. 
\end{abstract}

\noindent%
\vfill

\newpage
\spacingset{1.9} 
\section{Introduction}
\label{sec:intro}

Qualitative information is often considered less precise than its quantitative counterpart.
However, there are instances where qualitative information can be just as informative as quantitative information.  
To illustrate this, consider a scenario involving standardized test scores.  
Suppose we have a population of test takers, and their scores follow a normal distribution with a mean of 100 and a standard deviation of 10.  
Let us randomly sample one hundred test takers from this population. 
Instead of knowing the exact score of each test taker, we are only provided with their ranks.  
Despite not having the exact scores, using our knowledge of the normal distribution, we can determine that the fifth highest score is approximately 116.45, which is the mean plus 1.645 standard deviations. 
In this case, even though we only have qualitative information (the ranks), it is as informative as knowing the exact scores (quantitative information).  
This example highlights how qualitative information can provide valuable insights comparable to those obtained from quantitative information

In the given example, the statistical inference model comprises a prior probability distribution 
representing the measurement of interest, along with a side qualitative data. 
This inference model estimates the measurement from 
its quantitative prior conditioned on the side qualitative information. 
The ability to factor such qualitative information into a quantitative prior 
opens ways to combine human intuition into quantitative data models.

From this perspective, \cite{Almgren2006} and \cite{Chiarawongse2012} independently 
introduced a portfolio optimization model with qualitative input, 
which combines a qualitative view of an investor into a return estimation of financial assets.  
Subseqently, \cite{Eranda2021} extended the model to incorporate uncertainty in the qualitative view.  
The models by \cite{Chiarawongse2012} and \cite{Eranda2021} 
are derived from the vision of \cite{black:1991,black:1992} 
to fuse human view into quantitative data model.  
However, the view in the Black-Litterman model is also of a quantitative nature, and is 
deemed to be too demanding for human investors.
In a simulation study, \cite{Chiarawongse2012} reported that a qualitative input in the form of 
stock ranking when integrated with the expected return estimation can enhance the 
performance of a portfolio significantly.  This is especially true as the number of stocks increases. 
However, they did not provide any justification as to why such improvement in performance 
can be observed.  Our objective in this paper is to explain how such 
statistical model with qualitative input can perform so well.  
In the process, we also discuss the model properties and its limitation.  
We will see that dependence among variables compromises the performance of the model in 
an interesting way.

To investigate the properties of the proposed statistical model, one requires an estimation method.  
For this model, the estimation involves solving an integration problem over a convex polytope in high dimensions, which in general is intractable numerically \citep{Khachiy1989}. 
The solution method usually relies on a Markov chain Monte Carlo (MCMC) 
\citep{smith:1984, Kaufman1998, lovasz:2006:a, lovasz:2006:b, Kiatsupaibul2011}
in which, by allowing a confidence level lower than one, the method can solve the problem in polynomial time.

Even though the integration problem cannot be solved numerically in general, 
a numerical method can be devised to carry out the required integration for this particular model.  
Specifically, we consider a prior distribution that is multivariate normal with a uniform correlation structure 
together with qualitative information in the form of complete ranking. This specific problem is proved to 
be quite challenging for an MCMC.  However, \cite{Kiatsupaibul2017} proposed a solution method based on 
a recursive integration technique \citep{Hayter2006} that can be adapted to solve 
this specific problem in $O(n^2)$, where $n$ is the number of variables. 
In the example with ranking of stocks, $n$ is the number of stocks..  
In this paper, we present an adapted solution method based on that of \cite{Kiatsupaibul2017} 
to solve the problem.

The adapted solution method not only enables us to explore the properties of the model 
but also lends itself to practical use.  
An application of statistical inference with qualitative input can be found in the portfolio selection problem 
where statistical estimates of future expected returns are the main decision parameters \citep{Markowitz1952, Best1991, Chiarawongse2012}.  
With the solution method, we extend the simulation study in \cite{Chiarawongse2012}, 
by investigating the effects of the correlation coefficient on portfolio performance with perfect ranking information.
In addition to the simulation study, we also apply the proposed model to the portfolio selection problem with financial data from the US stock market. In the study, we observe the performances of the portfolios formed by the proposed model against a benchmark, when the ranking information is both perfect and imperfect.  When the information is imperfect, we show how the shrinkage estimator that incorporates confidence level into the model can improve the performance of the portfolio selection using the estimation based on data and ranking.

The organization of this paper is as follows.  In Section 2, the statistical inference model 
with qualitative input under investigation is defined.  The main result is that the asymptotic properties of 
the model is stated for a subclass where the prior distribution is normal with 
a uniform correlation structure and the qualitative input is in the form of ranking.  
In Section 3, the finite dimensional properties of the model are explored through the proposed recursive 
integration technique.  
In this section, the recursive integration technique adapted from \cite{Kiatsupaibul2017} is also described.  
In Section 4, the properties of the model with perfect qualitative information 
when applied to a portfolio selection problem is investigated.  
In Section 5, we exhibit the performance of the model when applied to the portfolio 
selection problem in an imperfect information setting.  In Section 6, a conclusion is provided.

\section{Statistical models with qualitative input}

Let $\vX=[X_1, \ldots, X_n]^\top$ be a normally distributed $n$ vector with mean vector $\vmu$ and 
covariance matrix $\vSigma$, i.e., $\vX\sim N(\vmu,\vSigma)$.  Let $\sR$ be a polytope defined by a 
set of linear inequalities.
\[\sR = \{\vx\in\Re^n\ :\ \vm{A}\vx\leq \vm{b}\},\]
where $\vm{A}$ and $\vm{b}$ are an $m\times n$ real matrix and a $m$ real vector, respectively.
The inference problem of concern is to compute the conditional expectation
\begin{equation}\label{eqn:estimate}
\ex[\vX\ |\ \vX\in \sR].
\end{equation}
The polytope $\sR$ represents a qualitative input, and this model can be interpreted as 
an estimation problem of a statistical quantity $\vX$ given that a qualitative information is available 
\citep{Chiarawongse2012, Kiatsupaibul2017}.

In this study, we limit ourselves to a particular model of $\sR$ and $\vX$.  
We restrict our attention to $\sR$ formed by a complete ranking, i.e.,
\begin{equation}\label{eqn:ranking}
\sR = \{\vx=[x_1, \ldots, x_n]^\top \in\Re^n\ :\ x_i\leq x_{i+1}, i=1,\ldots,n-1\}.
\end{equation}
We also call the model with a complete ranking $\sR$ as the \emph{rank constrained model}.
For the distribution of $\vX$, our study mainly concern the case where the correlation matrix possesses a uniform structure.

\begin{definition}\label{def:unifrho}
A normal random vector $\vX$ possesses a uniform correlation structure if the correlation coefficients between any pair of variables 
$X_i$ and $X_j$, for all $i,j\in\{1,\ldots,n\}$ and $i\neq j$, are all equal to a constant $0\leq \rho \leq 1$. 
In addition, if each variable $X_i$, for all $i=1,\ldots,n$, also has the standard normal (marginal) distribution, 
$N(0,1)$, we say that the vector $\vX$ possesses a standard uniform correlation structure.
\end{definition}

Observe that, if $\vX$ possesses a uniform correlation structure, it can be represented by a one-factor model as follows. 
Assume that $M$ and $Z_i, i=1,\ldots,n$ are independent and identically distributed (iid) random variables with distribution $N(0,1)$, 
\begin{equation}\label{eqn:onefactor}
X_i = \sigma_i(\sqrt{\rho} M +  \sqrt{1-\rho}Z_i) + \mu_i, \text{for } i = 1,\ldots,n,
\end{equation}
where $0\leq \rho\leq 1$.  In this representation, the random variable $M$ is a common 
factor that generates dependence among $X_i$'s.  
The random vector $\vX$ that possesses a standard uniform correlation structure can be represented by the one-factor model in 
(\ref{eqn:onefactor}) with $\mu_i=0$ and $\sigma_i=1$, for all $i=1,\ldots,n$.  
In what follows, for a set of random variable $X_1, \ldots, X_n$, we write their order statistics as 
\[X_{(1)},\ldots,X_{(n)}.\]

The following Theorem~\ref{thm:limqual} demonstrates that a ranking information has a potential to 
enhance the accuracy of a statistical estimate of a random vector in high dimensions. 
It also set a limitation of the estimation accuracy based on the dependence among variables.  

\begin{lemma}\label{lem:orderlimit}
Let $Z_1, \ldots, Z_n$ be iid $N(0,1)$.  Let $Z_{(\lceil np \rceil)}$ denote the 
$\lceil np \rceil$ order statistic of $Z_1, \ldots, Z_n$. Let $\Phi$ denote the distribution function of $N(0,1)$
and $\Phi^{-1}$ denote its inverse function.  We have, with probability one,
\begin{equation}\label{eqn:strongconv}
Z_{(\lceil np \rceil)} \to \Phi^{-1}(p).
\end{equation}
Furthermore, as $n\to\infty$,
\begin{equation}\label{eqn:exconv}
\ex[Z_{(\lceil np \rceil)}] \to \Phi^{-1}(p),
\end{equation}
and
\begin{equation}\label{eqn:varconv}
\var[Z_{(\lceil np \rceil)}] \to 0.
\end{equation}
\end{lemma}

\begin{theorem}\label{thm:limqual} 
For a normal random vector $\vX$ with standard uniform correlation structure 
and the complete ranking qualitative input $\sR$ in (\ref{eqn:ranking}), we have, for $0<p<1$, 
\begin{equation}\label{eqn:thmex}
\lim_{n\to\infty}\ex[X_{\lceil pn\rceil}\ |\ \vX\in \sR]=(1-\rho)\Phi^{-1}(p),
\end{equation}
\begin{equation}\label{eqn:thmvar}
\lim_{n\to\infty}\var[X_{\lceil pn\rceil}\ |\ \vX\in \sR]=\rho,
\end{equation}
\end{theorem}
\begin{proof}
Since $\vX$ possesses the standard uniform correlation structure, it can be represented by the one-factor model 
(\ref{eqn:onefactor}) with $\mu_i = 0$ and $\sigma_i = 1$ for all $i=1,\ldots,n$.  
The event $\vX\in \sR$ is equivalent to $Z_1 \leq Z_2 \leq \cdots Z_n$.  
Therefore, with $Z_{(i)}$ denoting the $i^\text{th}$ order 
statistic of $Z_i, i=1,\ldots,n$,
\begin{eqnarray*}
\lefteqn{\ex[X_{\lceil pn\rceil}\ |\ \vX\in \sR]}\\
 & = & \ex[\sqrt{\rho} M\mid \vX\in \sR] + \ex[\sqrt{1-\rho}Z_{\lceil pn\rceil}\mid \vX\in \sR] \\
 & = & \sqrt{\rho} \ex[ M\mid Z_1 \leq \cdots \leq Z_n] + \sqrt{1-\rho}\ex[Z_{\lceil pn\rceil}\mid Z_1\leq \cdots\leq Z_n] \\
 & = & \sqrt{\rho} \ex[M] + \sqrt{1-\rho}\ex[Z_{(\lceil pn\rceil)}] \\
 & = & \sqrt{1-\rho}\ex[Z_{(\lceil pn\rceil)}].
\end{eqnarray*}
(\ref{eqn:exconv}) imples (\ref{eqn:thmex}).
By a similar argument, with $M$ and $Z_i$'s being independent,
\begin{eqnarray*}
\lefteqn{\ex[X_{\lceil pn\rceil}^2\ |\ \vX\in \sR]}\\
 & = & \ex[\rho M^2 + 2\rho(1-\rho)MZ_{\lceil pn\rceil} + (1-\rho)Z_{\lceil pn\rceil}^2\mid \vX\in \sR] \\
 & = & \rho\ex[ M^2\mid \vX\in \sR] + 2\rho(1-\rho)\ex[MZ_{\lceil pn\rceil}\mid \vX\in \sR] \\ 
 & & \quad+(1-\rho)\ex[Z_{\lceil pn\rceil}^2\mid \vX\in \sR] \\
 & = & \rho\ex[ M^2\mid Z_1\leq\cdots\leq Z_n] + 2\rho(1-\rho)\ex[MZ_{\lceil pn\rceil}\mid Z_1\leq\cdots\leq Z_n] \\
 & & \quad+(1-\rho)\ex[Z_{\lceil pn\rceil}^2\mid Z_1\leq\cdots\leq Z_n] \\
 & = & \rho\ex[M^2] +2\rho(1-\rho)\ex[M]\ex[Z_{\lceil pn\rceil}\mid Z_1\leq\cdots\leq Z_n] +(1-\rho)\ex[Z_{(\lceil pn\rceil)}^2] \\
 & & \text{since $M$ and $Z_i$'s are independent,} \\
 & &  \rho\ex[M^2] +2\rho(1-\rho)\ex[M]\ex[Z_{(\lceil pn\rceil)}] +(1-\rho)\ex[Z_{(\lceil pn\rceil)}^2] \\
 & = & \rho\ex[M^2] +(1-\rho)\ex[Z_{(\lceil pn\rceil)}^2] \\
 & = & \rho + (1-\rho)\ex[Z_{(\lceil pn\rceil)}^2]
\end{eqnarray*}
Therefore,
\[
\var[X_{\lceil pn\rceil}\ |\ \vX\in \sR] = \rho + (1-\rho)\var[Z_{(\lceil pn\rceil)}].\\
\]
 By (\ref{eqn:varconv}), the last term on the right hand side of the last equation goes to zero as $n$ goes to infinity.  
 (\ref{eqn:thmvar}) then follows.
\end{proof}

In Theorem~\ref{thm:limqual}, Equation (\ref{eqn:thmex}) 
provides the first insight into the role of the dependence among variables to this model.  
When $\rho$ approaches one, Equation (\ref{eqn:thmex}) states that all expectations shrink towards zero.  
It suggests that the non-negative correlation coefficeint creates a clustering effect on 
the order statistics of the variables of interest.  

Another role of the dependence among variables, which is crucial for our purpose, 
is given in Equation (\ref{eqn:thmvar}).
In this model, we estimate the expectation of each random variable $X_i$ based on a ranking information 
by its conditional expectation given the ranking. Therefore, the conditional variance given the ranking measures the estimation accuracy, 
and (\ref{eqn:thmvar}) specifies the limiting accuracy of the estimate. 
From (\ref{eqn:thmvar}), when $\rho=0$, the conditional variance given the ranking is zero in the limit, 
implying that, in the limit, the conditional expectation given the ranking is a perfect estimate with no error.  
Also from (\ref{eqn:thmvar}), the limiting accuracy deteriorates when $\rho$ increases, suggesting that the 
dependence among the variables is the major source of the estimation error to this model when the number of variables is large.  
In conclustion, under the standard uniform correlation structure, at small $\rho$ and with $n$ large, 
the estimate from the conditional expectation given the ranking achieves high accuracy, even though the information that forms the estimate is only of qualitative nature.  

Theorem~\ref{thm:limqual} only specifies the limiting behavior of the estimate of the the inference 
model with ranking input.  To study the finite dimensional behavior of this model, 
a computation method for the conditional expectation given a ranking (\ref{eqn:estimate})
at finite but large $n$ is required. 
This computation would also enable this model to be deployed in real world applications.  
In the next section, we describe a computational method to perform this task.

\section{A Computational method}\label{sec:method}

In this section, a numerical integration method for evaluating the conditional expectation 
(\ref{eqn:estimate}) is introduced.  
With the computation method, we study the finite dimensional behavior of the inference model 
with the standard and non-standard uniform correlation. 

The conditional expectation (\ref{eqn:estimate}) requires an $n$-dimensional integration operation.  
\cite{Kiatsupaibul2017} provides a recursive integration method that reduces this $n$-dimensional integration to 
a series of two-dimensional integration operation.  
The following is the recursive integration formula by \cite{Kiatsupaibul2017} that is adapted to the one-factor model (\ref{eqn:onefactor}).

\begin{equation}\label{eqn:condexform1}
\ex[X_k\mid \vX\in \sR] = \frac{A}{B}
\end{equation}
where $A$ and $B$ are evaluated by the following two-dimensional recursive integration formulae.
Let $\phi(x)$ denote the probability density function (pdf) of the standard normal distribution $N(0,1)$.
For each $m\in\Re$, let $\phi_i(m,x)$ denote pdf of the normal distribution 
$N\left(\mu_i + \sigma_i\sqrt{\rho}m, \sigma_i^2(1-\rho)\right)$.  
\begin{equation}
B = \int_{-\infty}^\infty  \int_{-\infty}^\infty \phi(m) b_{n-1}(m,x)\phi_n(m,x)\,dx\,dm,
\end{equation}
where $b_0(m,x)=1$ and, for $i=1,\ldots,n-1$, define for each $m,x\in\Re$,
\begin{equation}
b_i(m,x) = \int_{-\infty}^x b_{i-1}(m,t)\phi_i(m,t)\,dt.
\end{equation}
The recursive integration formula for $A$ is as follows.  For $i=1,\ldots,n$ and $i\neq k$, define $g_i(x) = 1$, and when $i=k$,
\[g_k(x) = x,\]
where $k$ is the index of the variable in (\ref{eqn:condexform1}) whose expectation to be evaluated.
\begin{equation}
A = \int_{-\infty}^\infty \int_{-\infty}^\infty \phi(m)h_{n-1}(m,x)g_n(x)\phi_n(m,x)\,dx\,dm,
\end{equation}
where $h_0(m,x) = 1$ and, for $i=1,\ldots,n-1$, define for each $m,x\in\Re$,
\begin{equation}
h_i(m, x) = \int_{-\infty}^x h_{i-1}(m,t)g_i(t)\phi_i(m,t)\,dt.
\end{equation}
The conditional second moment given the ranking can be with the same formula by replacing $g_i(x) = x^2$ when $i=k$.  
The variance and standard deviation given the ranking can be deduced from the conditional expectation and 
the conditional second moment.
(Refer to \cite{Kiatsupaibul2017} for the implementation and some properties of the above recursive integration formulae.)

With the inference method described above, we can study the finite dimensional behavior of the rank constrained 
inference model with the standard uniform correlation structure.  
Figure~\ref{fig:convspeed} show the convergence speeds of the standard deviations of the estimates to the limit implied by Theorem~\ref{thm:limqual}.

\begin{figure}[h!]
\begin{center}
\includegraphics[scale=0.5]{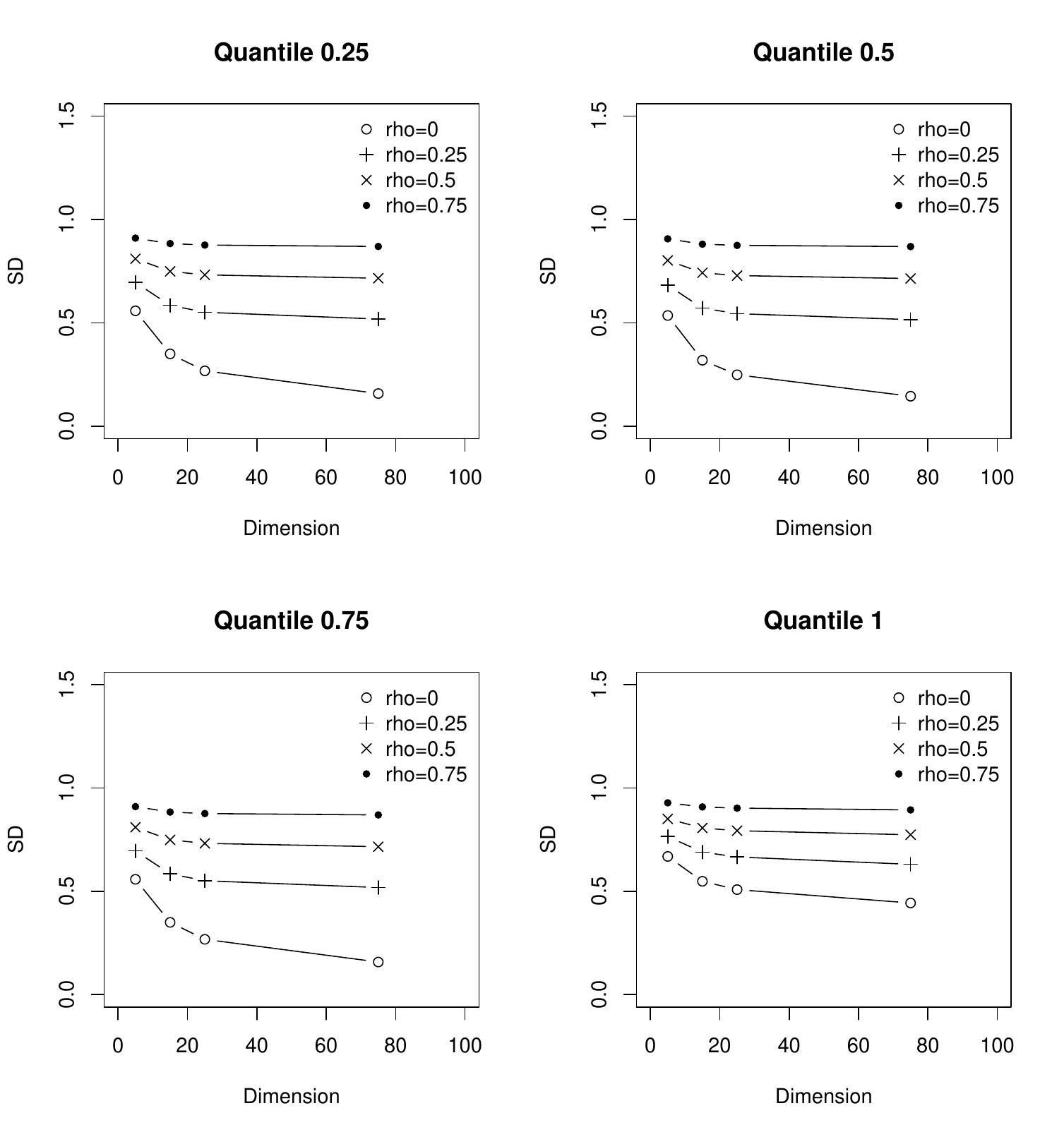}
\end{center}
\caption{Convergence speeds of the SDs at different quantiles and $\rho$\label{fig:convspeed}}
\end{figure}

Consider the case when $\vX$ possesses a standard uniform correlation structure
with $\mu_i = 0$ and $\sigma_i = 1$ for all $i=1,\ldots,n$.  
We compute the conditional standard deviation of $\vX$ given 
the complete ranking information at dimensions $n = 5, 15, 25, 75$
where the correlation coefficient $\rho$ are controlled at $\rho=0, 0.25, 0.50, 0.75$.
The computation are performed by recursive integration technique described above.
The conditional standard deviations of the 0.25, 0.50, 0.75 and 1.00 quantiles of the estimates are shown in 
Figure~\ref{fig:convspeed}. 
From Figure~\ref{fig:convspeed}, at each $\rho$ and quantile, 
we observe that the conditional standard deviation (SD) decreases as the dimension $n$ grows.  
The speed of reduction in the conditional SD is tapered off at high dimension $n$.
Since the conditional SD measures the estimation accuracy, this result implies that the estimation accuracy 
increases at higher dimensions, but converges to the limit imposed by Theorem~\ref{thm:limqual}. 
We also observed that the conditional SD is smaller with smaller correlation coefficient $\rho$. 
This result implies that the estimation accuracy is compromised by the dependence among the variables. 

From Figure~\ref{fig:convspeed}, at a fixed $\rho$, there is no obvious difference in the graphs among 
0.25, 0.50 and 0.75 quantiles.  However, the graph for 1.00 quantile is quite different from the others.  
It should be noted that the limiting conditional SD for 0.25, 0.50 and 0.75 quantiles are all governed by 
Theorem~\ref{thm:limqual}, which states that they all converge to a constant.  However, the limiting 
conditional SD for 1.00 quantile is beyond the scope of Theorem~\ref{thm:limqual}.  When $\rho=0$, 
the limit of the conditional SD for 1.00 quantile is governed by the extreme value theorem.  Therefore, 
it is possible that the decreasing pattern for the conditional SD in the case of 1.00 quantile is different from 
those in the other cases.  For the case of 1.00 quantile with $\rho \neq 0$, there is no limit theorem 
to explain the limiting behavior of the conditional SD.  Nevertheless, in Figure~\ref{fig:convspeed}, we 
can still observe the decreasing pattern of the conditional SD for 1.00 quantile with $\rho\neq 0$ in 
finite dimensions.

Now consider a case when $\vX$ possesses a non-standard uniform correlation structure, i.e., 
$X_i$ are not identically distributed.  
Let $\sigma_i=1$, for $i=1,\ldots,n$, but 
$\mu_i$'s be different from one another.  
We call a qualitative input $\sR$ reinforcing if $\vmu\in \sR$.  
On the other hand, if $\vmu\notin \sR$, we call $\sR$ opposing.  
The degree of reinforcement depends on how deep $\vmu$ is in $\sR$.  
We would like to observe the effect of the degree of reinforcement 
on the accuracy of the estimate (\ref{eqn:estimate}).  

To do so, we set $\vmu$ as follows.  Let vector $\vnu$ be a 
vector whose component $i$ is
\[\nu_i = -1 + \frac{2(i-1)}{n-1}, i=1,2,\ldots,n\]
Then $\vnu\in \sR$ whose components are equi-spaced.  
Now let $\vmu$ be
\begin{equation}
\vmu = r\left(\frac{\vnu}{\|\vnu\|_2}\right),
\end{equation}
where $\|\vnu\|_2$ is the Euclidean distance of $\vnu$.  
In other words, $\vmu$ is the equi-spaced vector that is scaled to have length $r$, 
emanating from the origin, which is the tip point of the cone $\sR$.  
The length $r$ can be regarded as the degree of reinforcement.  
It should be noted that $r$ can be negative.  
A negative length $r$ expresses the degree of opposition of $\vmu$ to the input $\sR$.
We call $r$ the reinforcement index.  Figure~\ref{fig:reinforcingrho0} shows the conditional SD 
of the estimates versus $r$ when $\rho$ is 0.

\begin{figure}[h!]
\begin{center}
\includegraphics[scale=0.5]{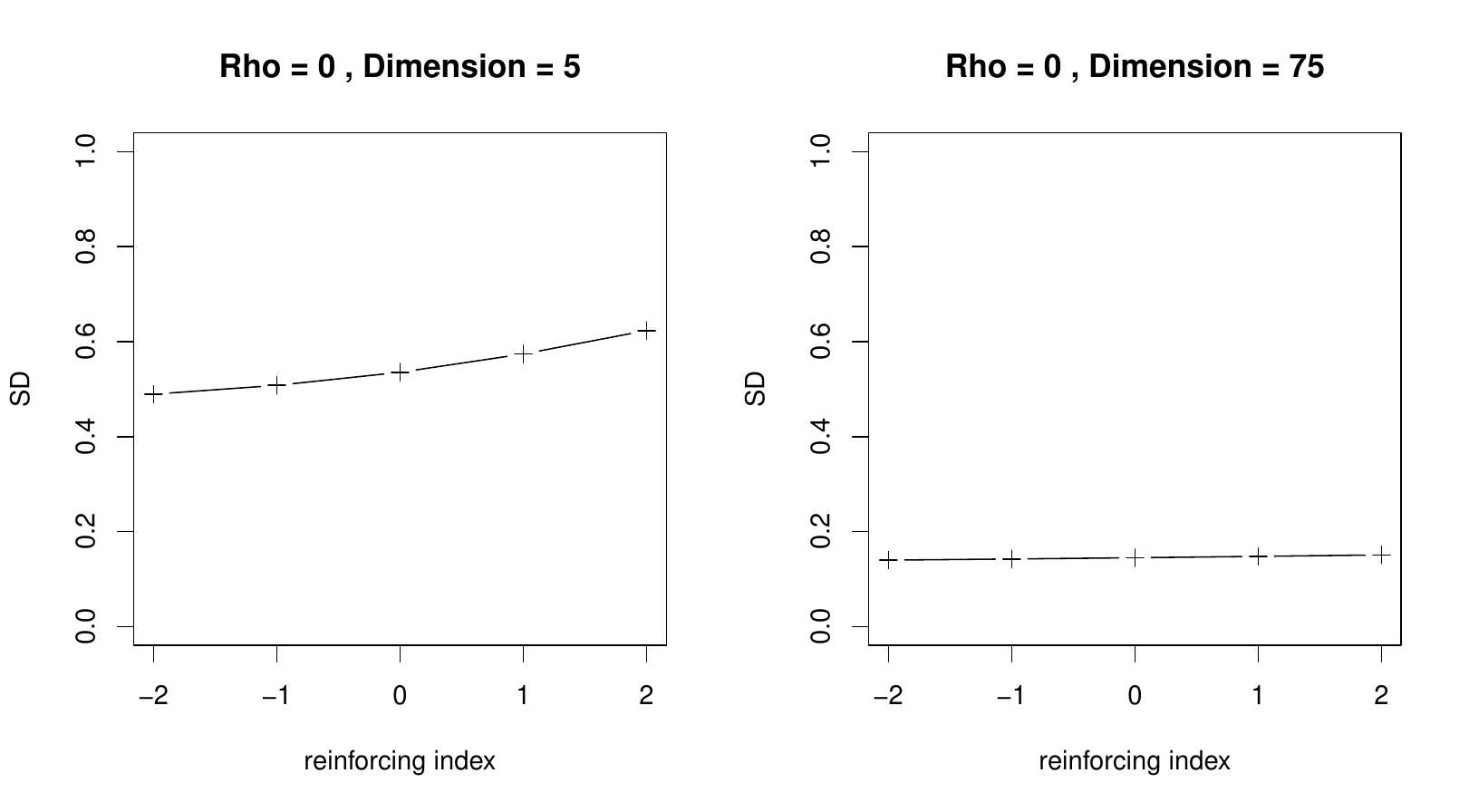}
\end{center}
\caption{The accuracy of the estimates measured by the standard deviation of 
the 50th sample percentile at $\rho=0$ when $n=5$ and $75$\label{fig:reinforcingrho0}}
\end{figure}

From the left panel of Figure~\ref{fig:reinforcingrho0}, in low dimensions ($n=5$), 
the graph of the conditional SD is tilted upward from a negative $r$ to a positive $r$. 
This increasing pattern implies that, in low dimensions, 
we obtain a higher accuracy of the estimates when we have opposing inputs.  
However, in higher dimensions ($n=75$), as shown in the right panel of Figure~\ref{fig:reinforcingrho0}, the 
graph is relatively flat.  This pattern implies that, in high dimensions,
there are no differences in the accuracy between the reinforcing inputs and the opposing inputs.  
Figure~\ref{fig:reinforcingrho0_5} shows the conditional SD of the estimates versus $r$ when $\rho$ is 0.5.  
From the left panel of Figure~\ref{fig:reinforcingrho0_5}, when $\rho=0.5$, 
we still observe the increasing pattern of the graph in low dimensions ($n=5$) even though it is not as pronounced 
as when $\rho=0$.  From the right panel of Figure~\ref{fig:reinforcingrho0_5}, when $\rho=0.5$
the graph of the conditional SD is relatively flat in higher dimensions ($n=75$), similar to the case
when $\rho=0$.

\begin{figure}[h!]
\begin{center}
\includegraphics[scale=0.5]{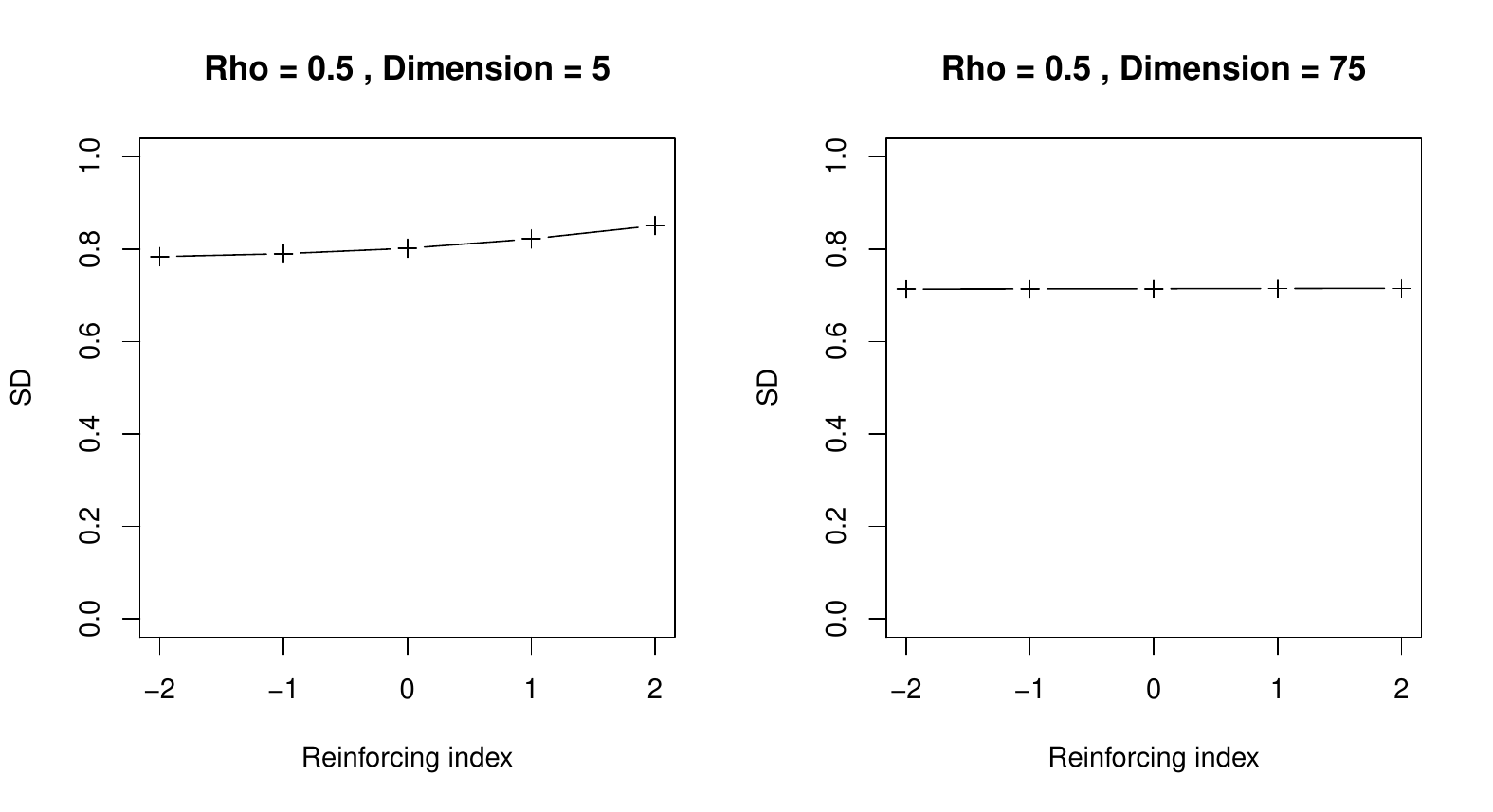}
\end{center}
\caption{The accuracy of the estimates measured by the standard deviation of 
the 50th sample percentile at $\rho=0.5$ when $n=5$ and $75$\label{fig:reinforcingrho0_5}}
\end{figure}

\section{An application to portfolio selection}

An application of the statistical estimation with ranking input can be found in a mean-variance portfolio selection problem.  
In a mean-variance portfolio context, the random variables of interest are, from a Bayesian perspective, 
the future expected returns on the assets. 
That is, in a universe of $n$ assets, $X_i$ is the future expected return on asset $i$ for $i = 1,\ldots, n$.   
\cite{Best1991} showed that the solution to the portfolio selection problem is very sensitive to the estimate 
of $X_i$.  From Theorem~\ref{thm:limqual}, we learned that ranking information can enhance the 
accuracy of the estimates of $X_i$ at high dimensions.  Consequently, the performance of a portfolio 
would also be enhanced when the accuracy of the estimates of $X_i$ is improved by the ranking input.   
In Theorem~\ref{thm:limqual}, we see that the degree of the accuracy of the estimates with ranking input 
is controlled by the correlation coefficient $\rho$. In this section, we investigate the effects of $\rho$ on the performance of the portfolio formed by the estimates 
of $X_i$ from the methodology put forth.

Recall a mean-variance portfolio selection problem: 
\begin{equation}\label{eqn:markowitz}
\max_{\vone^\top \vw=1}\ 
{\hat{\vmu}}^\top \vw-\frac{\gamma}{2}\vw^\top{\hat{\vXi}}\vw.
\end{equation}
To determine an optimal portfolio weight vector $\vw$, this problem requires, as inputs, $\hat{\vmu}$, the estimate of future expected asset return vector $\vX$, and $\hat{\vXi}$, the estimate of the covariance matrix of future asset returns. 
In order to single out the effect of the estimation accuracy on the optimal portfolio performance, we assume that 
$\hat{\vXi} = \vXi$ is known, and we estimate only $\hat{\vmu}$.  
In this section, we evaluate the performances of portfolios constructed based on the return estimates with ranking input 
at different levels of $\rho$. We then compare the performances of portfolios with rank constrained return estimates 
against some benchmarks.  
We separate the study into three cases based on the covariance matrix structure, 
namely the standard uniform correlation structure, the uniform correlation structure and the general one factor model.  
In the standard uniform correlation structure case, the performance of portfolios based on rank constrained return estimates 
can be assessed through Theorem~\ref{thm:limqual}.  In the non-standard uniform correlation structure case, 
the portfolio performance is assessed  through a simulation experiment.  In the general one factor model, the portfolio 
performance is investigated based on a real financial data set.

\subsection{Standard uniform correlation structure}

We first consider the scenario where the future expected asset returns possess a standard uniform correlation structure.  
In Theorem~\ref{thm:limqual}, we have seen that the correlation coefficient $\rho$ reduces the accuracy of 
a rank constrained estimate of the future expected asset returns. 
However, one can show that the performance of the portfolio is not affected by $\rho$.  
Consider the following example.

\begin{example}\label{ex:shift}
Consider two portfolio selection problems with two different return estimates $\hmu$ and $\cmu$.  Let $\hmu$ 
and $\cmu$ be different by a constant, i.e., 
\[\hmu = \sqrt{\rho}\hat{\vm{m}} + \sqrt{1-\rho}\,\vm{z} \text{ and } \cmu = \sqrt{\rho}\,\check{\vm{m}} + \sqrt{1-\rho}\,\vm{z}, \]
where $\hat{\vm{m}}$ and $\check{\vm{m}}$ are constant vectors of values $c_1$ and $c_2$, respectively.  
Let us assume that the other parameters, which are the covariance matrix and the risk aversion parameter, are the same for the 
two portfolio problems.  One can easily see that the optimal solutions to the two portfolio problems are the same.  
That means the optimal solution will not be affected by a parallel shift of the return input.  
\end{example}

In the discussion prior to Theorem~\ref{lem:orderlimit}, 
a future expected return vector with a standard uniform correlation structure can be written in a 
one-factor model defined in (\ref{eqn:onefactor}).
The estimation error of the return vector can be decomposed into the estimation error of the common factor $\vM$ 
and that of the idiosyncratic term $\vZ$.
From the proof of Theorem~\ref{thm:limqual}, the estimation error of 
the idiosyncratic term $\vZ$ is eliminated by the knowledge of a perfect ranking at high dimensions, 
while the estimation error of the common factor
$\vM$ remains intact.  However, the estimation error from the common factors $\vM$ is only a parallel shift in the return estimation. As shown in Example~\ref{ex:shift}, the parallel shift does not influence the optimal solution. 
Therefore the increase in the correlation efficient, even though reduce the estimation accuracy, 
does not compromise the portfolio performance. 

\subsection{Uniform correlation structure}\label{subsec:uniformcorr}

The characteristics of a rank constrained return estimates and, hence, the performance of a portfolio, 
with respect to a non-standard uniform correlation structure, is beyond the scope of Theorem~\ref{thm:limqual}.
To assess the effect of correlation coefficient on the portfolio performance, we resort to a simulation experiment 
equipped with numerical computation method described in Section~\ref{sec:method}. We extend the simulation study of \cite{Chiarawongse2012} by controlling $\rho$, the correlation coefficient defined in (\ref{eqn:onefactor}). 
We carry out the mean-variance portfolio selection in (\ref{eqn:markowitz}), assuming that 
$\hat{\vXi} = \vXi$ is known and estimating only $\hat{\vmu}$. 
The objective of the simulation is to compare the optimal portfolio performance with 
respect to three types of $\hat{\vmu}$ estimates: the prior mean, the true mean and 
the conditional mean with ranking input.

In our simulation study, we assume that the future expected asset return vector $\vX$ has a uniform correlation structure.
To simulate the vector $\vX$, we require the following hyperparameters: prior mean vector $\vmu$, 
the standard deviation parameter for generating the prior means $\sigma_{\mu}= 2.5\times 10^{-7}$, the variance parameter for generating the 
covariance matrix ${\sigma}^2_\Sigma = 1\times 10^{-3}$, the scaling parameter for covariance matrix $\tau = 0.1$, and 
the correlation coefficient $\rho$. 
The sequence of steps in our simulation is as follows.
\begin{enumerate}[leftmargin=2cm, label=\emph{Step }{\arabic*}]
\item Simulate an $n$-vector $\vmu$ of prior means whose components $\mu_i, i=1,\ldots,n$ 
are iid and each one has $N(0,{\sigma}^2_\mu)$ distribution.
\item Simulate covariance matrix of asset returns $\vXi$ as follows.  
First, we simulate an $n$-vector $\vm{s}$ of scaling factors whose component $s_i, i=1,\ldots,n$ are iid and each one has $\chi^2_n$ distribution.
We then let $\vXi = \vm{S}\vm{D}\vm{S}$ where $\vm{S}$ is an $n\times n$ diagonal matrix whose $i^{\text{th}}$ diagnonal entry is 
$\sqrt{s_i{\sigma}^2_\Sigma}$ and $\vm{D}$ is an $n\times n$ correlation matrix whose off-diagonal entries all equal 
the correlation coefficient $\rho$.
\item Simulate $\vX\sim N(\vmu,\vSigma)$ where $\vSigma=\tau\vXi$.
\item Extract the ranking information $\sR$ from $\vX$.
\item Form the three portfolios based on the three estimations of $\hat{\vmu}$ and measure their performances. 
\end{enumerate}
The experiments are done in this setting where the correlation coefficient $\rho$ 
and the number of assets $n$ are controlled at 
\begin{itemize}
\item $\rho= 0, 0.25, 0.50, 0.75$,
\item $n = 5,15,25,75 $.
\end{itemize}

In \emph{Step} 5, we execute three mean-variance portfolio selection models based on the three types of future expected return estimate 
substituted into $\hat{\vmu}$ in (\ref{eqn:markowitz}):
\begin{itemize}
\item \emph{Prior}: The prior mean $\vmu$.
\item \emph{Clairvoyance}: The true $\vx$ that is a sample of the future expected asset return vector $\vX$ 
generated by the simulation.
\item \emph{Rank constrained}: The conditional mean with ranking information 
\begin{equation}\label{eqn:mutilde}
\tilde{\vmu}=\ex[\vX\mid \vX\in \sR].
\end{equation}
\end{itemize}

For the rank constrained model (\ref{eqn:mutilde}), we use $\sR$ as the perfect ranking of the 
future expected asset returns extracted from the simulated $\vX$.  
The conditional expectation (\ref{eqn:mutilde}) with ranking information $\tilde{\vmu}$ is computed 
from the prior mean $\vmu$, the covariance matrix $\vSigma$ and the ranking information $\sR$ 
by the recursive integration technique described in Section~\ref{sec:method}.

The optimal-weight vector for the three portfolio selection models are computed according to (\ref{eqn:markowitz}) 
by adopting the true $\vXi$ as the covariance matrix.  Following \cite{Chiarawongse2012}, 
the risk aversion parameter $\gamma$ is set to 4.
We denote the optimal weight vectors corresponding to the prior model, the clairvoyance model and the rank constrained 
model by $\bar{\vw}$, $\vw^*$ and $\tilde{\vw}$, respectively.

We solve 100 instances of the portfolio selection problem (\ref{eqn:markowitz}) with simulated parameters obtained by the simulation environment 
described above and with different estimates of $\hat{\vmu}$.  
In each instance, we evaluate the performance of each portfolio selection model based on 
the Certainty Equivalence Return (CEQ) defined as 
\begin{equation}\label{eqn:CEQ}
\ \text{CEQ} = {\vx}^{T}\vw-\frac{\gamma}{2}{\vw}^{T}{\vXi}\vw,
\end{equation}
where $\vx$ is the sample of $\vX$ and the weights vector $\vw$ 
varies according to each model solution ($\bar{\vw}$, $\vw^*$ or $\tilde{\vw}$).
Finally, we average the performances of the 100 instances of the 
three portfolio models at the different values of correlation coefficient 
and compare the average performance across number of assets $n$ as shown in Figure~\ref{fig:port1}. 
To facililitate the comparison across levels of correlation coefficient, we also compute the percent differences 
between the clairvoyance and the rank constrained model as shown in Figure~\ref{fig:port2}.

\begin{figure}[h!]
\begin{center}
\includegraphics[scale=0.5]{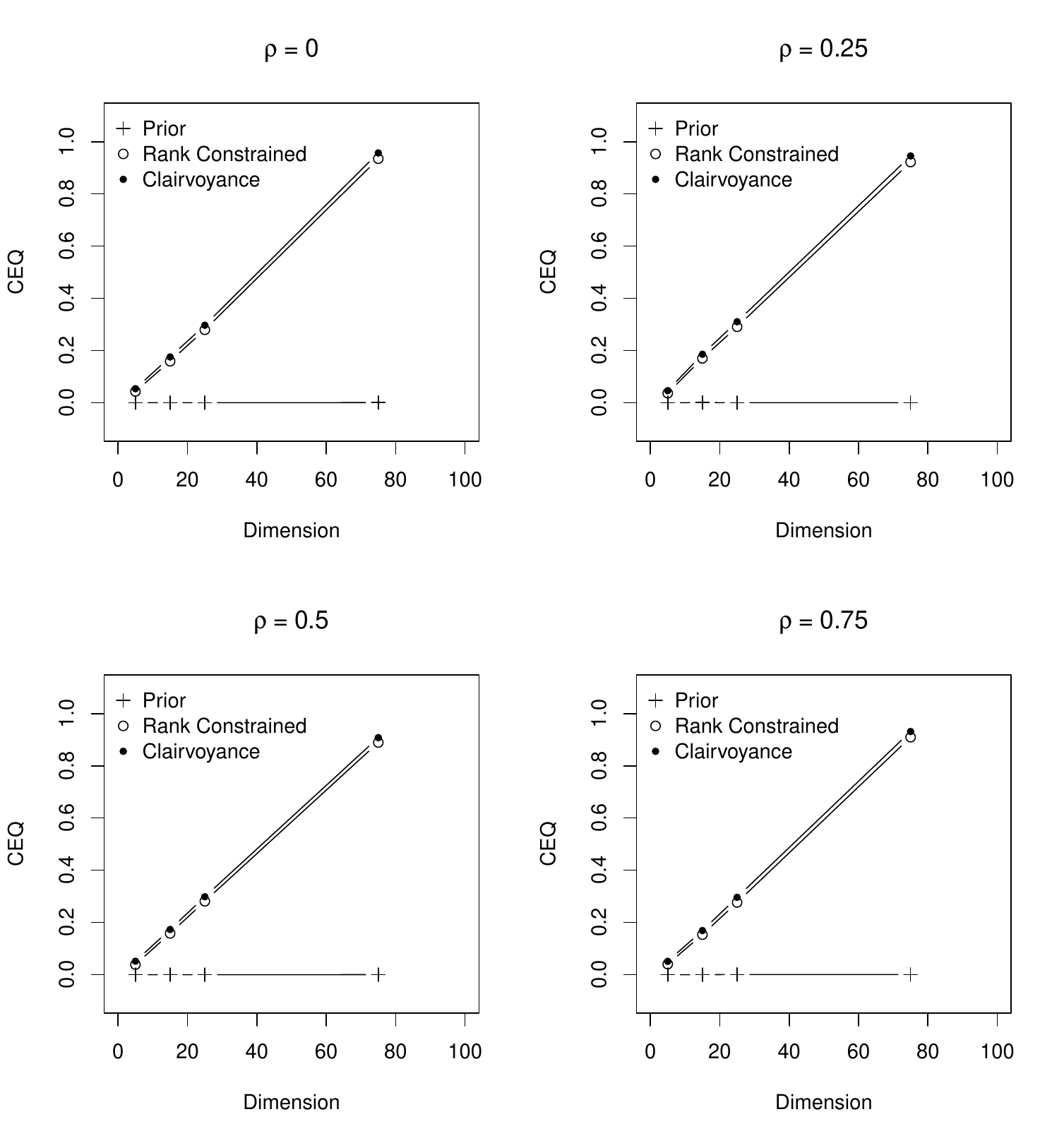}
\end{center}
\caption{Performances of different portfolio selection models based Certainty equivalence (CEQ) as functions of the 
number of assets (Dimensions).  Different panels shows the performances versus number of assets at 
different values of correlation coefficient.}\label{fig:port1}
\end{figure}

\begin{figure}[h!]
\begin{center}
\includegraphics[scale=0.5]{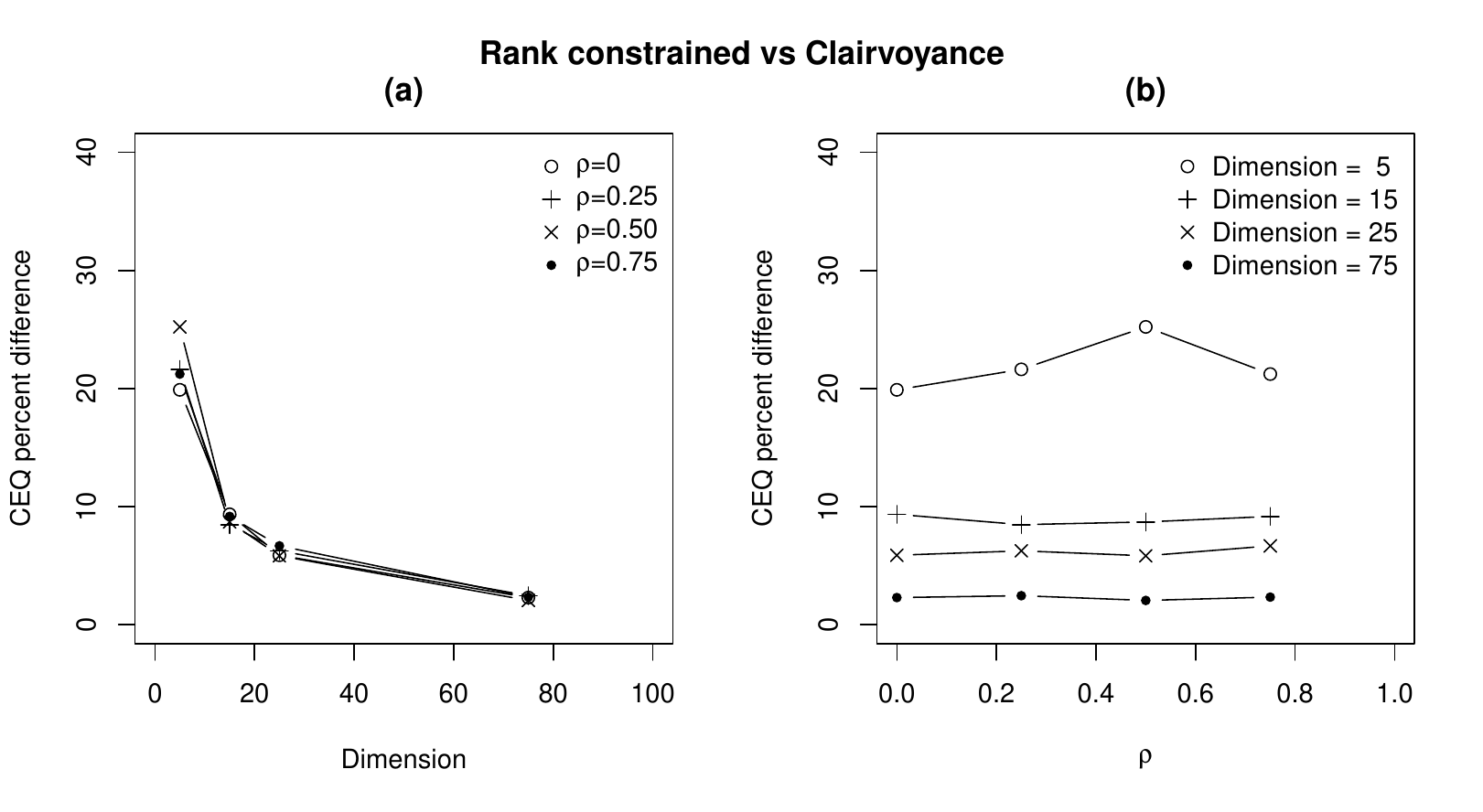}
\end{center}
\caption{Panel (a) shows the percent differences between the clairvoyance and the rank constrained model as functions of the dimensions.  Panel (b) shows the percent differences between the clairvoyance and the rank constrained model as functions of the correlations coefficients.}\label{fig:port2}
\end{figure}

Figure~\ref{fig:port1} shows the performance evaluated for each model at different numbers of assets $n = 5,15,25,75$.
For each return estimation model, for all levels of correlation coefficient $\rho$, the clairvoyance model achieves the highest performance while trailed closely by the rank contrained model. 
The prior model leads to the worst performance in every setting. 
This confirms the benefit of incorporating ranking information in the return estimation as previously found in \cite{Chiarawongse2012}.

Figure~\ref{fig:port2} shows the percent differences between the clairvoyance model and the rank constrained model 
as functions of the dimension in Panel~(a) and as functions of the correlation coefficients in Panel~(b). 
According to Figure~\ref{fig:convspeed}, the estimation discrepancy, represented by the standard deviation of the estimator, 
declines in higher dimensions.  
Panel~(a) of Figure~\ref{fig:port2} shows that the performance discrepancy between the ranked constrained model and the clairvoyance model also declines in a larger portfolio as expected. 
From Figure~\ref{fig:convspeed}, since the standard deviation of the estimator with ranking information increases 
with larger correlation coefficient $\rho$, one may expect the performance of the rank constrained model to deteriorate 
commensurately with larger $\rho$.  However, as seen in the case of the standard uniform correlation structure, 
the return estimation error from $\rho$ is largely a result of a parallel shift, which does not compromise the performance 
of the portfolio.  
Panel~(b) of Figure~\ref{fig:port2} shows no visible trend of the 
relative portfolio performance between the rank constrained and the clairvoyance models as $\rho$ grows larger.  
This result confirms that ranking information when fused into the return estimations can eliminate the influence of $\rho$ on 
portfolio performance.

\subsection{One-factor Model}\label{sec:onefactor}

In this section, 
we investigate the portfolio performance with a more general correlation structure given a perfect ranking information.  
We assume that the future expected return vector $\vX$ can be well represented by the one factor model.
\begin{equation}\label{eqn:onefactorgen}
X_i = \sigma_i(\lambda_iM +  \sqrt{1-\lambda_i^2}Z_i) + \mu_i, \text{for } i = 1,\ldots,n,
\end{equation}
where $M, Z_i\sim N(0,1)$ for $i=1,\ldots,n$ are iid and $0\leq \lambda_i \leq 1$ for $i=1,\ldots,n$ are the factor loadings.
We employ real financial time series from Kenneth French’s web site at Dartmouth. 
We approximate the correlation structure of this real 
data set by the one-factor model, which is a common practice in portfolio selection methodology.  
In addition, this approximation also lends itself to the computation method we discussed. 
The perfect ranking information is hypothetically generated from 
the ranking of the target variables one-step ahead in the financial time series.  
The objective of this section is to compare the performance of the portfolios with parameters obtained 
from the rank constrained statistical model with the one factor correlation structure and the perfect ranking 
against alternatives.

The financial time series under this study contains 
ten industry monthly asset returns during the period of 07/2005 - 12/2022.  
We follows the rolling-sample approach in \cite{DeMiguel2009}.
In each monthly time step from 01/2011 to 12/2022, 
we form a portfolio based on the rank constrained model, 
resulting in the total of 144 portfolios that are re-adjusted monthly.  
We the repeat the process to obtain 144 benchmark portfolios to compare against those from the rank constrained model.
The covariance matrix for the target month $t$ is estimated by applying the one factor analysis to the data in the period 
of 66 months prior to the target month $t$, while the perfect ranking is taken from the ranking of the asset returns 
of the target month.  
For example, for the portfolio formed for month 01/2011, the covariance matrix is estimated from months 07/2005 to 12/2010, 
while the ranking is taken from the target month 01/2011. We denote the estimated covariance matrix of month $t$ by 
$\hat{\vXi}_t$.

To estimate the rank constrained expected return $\tilde{\vmu}_t$ according to \eqref{eqn:mutilde}, 
the prior mean $\vmu_t$ and the prior covariance matrix $\vSigma_t$ of the expected return at month $t$ are required. 
Following \cite{black:1991,black:1992}, the prior mean $\vmu_t$ of period $t$ is implied from 
the market portfolio of the last time period $t-1$.  
The market portfolios are also obtained from Kenneth French’s web site at Dartmouth.
Following Section~\ref{subsec:uniformcorr}, 
the prior covariance matrix $\vSigma_t$ is taken as $\tau \hat{\vXi}_t$ where $\hat{\vXi}_t$ 
is the estimated covariance matrix of month $t$ and $\tau$ is set to be equal to 0.1.
Since $\hat{\vXi}_t$ is a result of the one-factor model, the estimated rank constrained return $\tilde{\vmu}_t$ 
can be computed by the recursive integration method in Section~\ref{sec:method} where $\rho$ is replaced by 
the loading factor from the one factor analysis.  See \cite{Kiatsupaibul2017} for details.  

In this section, we compare the out-of-sample empirical performances of the rank constrained portfolios with 
the perfect ranking to the benchmark portfolios using 
the estimated certainty equivalence return (CEQ) as the performance metric. 
Following the rolling-horizon experiments, 
we construct the rank-constrained portfolios according to the unconstrained optimization model as in \cite{DeMiguel2009},
using rank constrained expected return $\tilde{\vmu}_t$ 
and the estimated covariance matrix $\hat{\vXi}_t$. 
The benchmark portfolio of month $t$ is then computes as the market portfolio 
where we assigns a weight to each asset equal to the market capitalization of that asset 
in month $t-1$ divided by the total market capitalization of month $t-1$.
Holding the portfolio for one month gives the out-of-sample return at time $t$. 
For each type of portfolio, we then partition the time series of portfolio returns of 
$144$ months into $12$ consecutive twelve-month periods. 
In each period $l$ for $l=1,\ldots,12$, we define the estimated CEQ as 
\[\hat{\text{CEQ}}_l = \bar{r}_l - s_l^2/2,\quad l=1,\ldots,12,\]
where $\bar{r}_l$ and $s_l^2$ are the average and 
the sample variance of the monthly returns over period $l$.  

The estimated CEQs of the rank-constrained model and the benchmark are shown as 
the dotted line and the solid line in Figure~\ref{fig:emport1}, respectively.  
The CEQ's from the rank constrained model with the perfect ranking (dotted line) 
consistently outperform those from the benchmark model (solid line).  
This result suggests that the rank constrained model can effectively capitalized on 
the value of ranking information when the information is perfect.

\begin{figure}[h!]
\begin{center}
\includegraphics[scale=0.60]{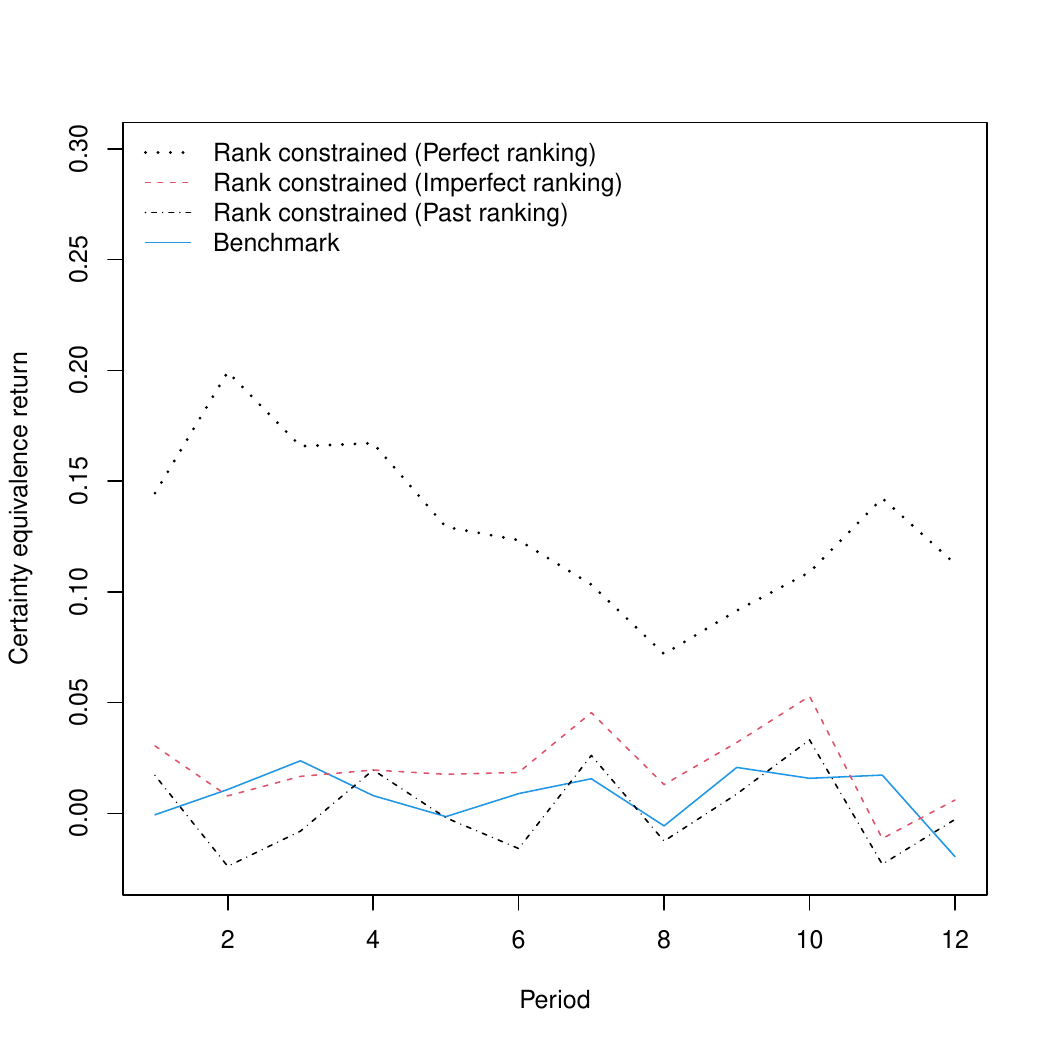}
\end{center}
\caption{Certainty equivalence returns (CEQ) of the 12 twelve-month periods from the 
rank constrained models with the perfect ranking (dotted line), the past rankings (dot-dashed line), 
the uncertain ranking (dashed line) 
and the benchmark model (solid line).  
}\label{fig:emport1}
\end{figure}

\section{Portfolio with imperfect qualitative information}

In the previous sections, we assume that the qualitative information in the form of a complete ranking is perfect.  
In practice, ranking information is hardly perfect.  
We repeat the experiments in Section~\ref{sec:onefactor} with imperfect ranking information.
To simulate an imperfect ranking, 
we observe past ranking and mix them with the perfect ranking.
In this study, the past ranking is defined as the ranking of the monthly asset returns one month prior to the target month.  
For example, a portfolio with rank constrained model with the past ranking for month 01/2011 
is formed as that with perfect ranking except that the ranking that is used in the estimation process is the 
ranking of the asset returns of month 12/2010.  
The performances of the portfolios with the past ranking are shown by the dot-dashed line in Figure~\ref{fig:emport1}.  
Observe that the portfolios with the past ranking (dot-dashed line) do not consistently outperform 
the benchmark (solid line), suggesting that the past ranking offers no obvious benefit to the portfolio selection process.

To additionally investigate a more practical aspect of the model, we observe the performances of 
the portfolio based on imperfect rankings.  To form a portfolio with an imperfect ranking for a target month, 
we follow the same process as that of the perfect ranking, 
but replace the perfect ranking with a ranking randomly chosen between the perfect ranking and the past ranking. 
In particular, to form a portfolio of a target month $t$, with probability 0.25, we employ the perfect ranking in the 
estimation process, and, with probability 0.75, we employ the past ranking.  
The performances of the portfolio with these imperfect rankings are shown by the dashed line in Figure~\ref{fig:emport1}.  
Observe that the performances of the portfolios with the imperfect rankings are mostly higher 
than those of the benchmark portfolios but with smaller margins than those corresponding to the perfect ranking.  
This suggests that, even with an imperfect ranking, benefits from the model can still be observed.  

\section{Conclusion}

Incorporating qualitative input in the form of ranking has the potential to enhance the accuracy of statistical estimates for variables. 
This improvement is particularly notable in high-dimensional settings when the variables of interest are independent and follow a standard normal distribution. According to Theorem~\ref{thm:limqual}, the statistical estimates obtained using ranking input can be exact in such cases. 
In scenarios with a moderate number of dimensions, there exists an efficient numerical algorithm that computes these estimates, enabling their practical utilization. 
This computational tool also facilitates an examination of the convergence speed of the estimates and the impact of the degree of reinforcement in the input. However, when the variables exhibit a uniform correlation structure, the accuracy of the estimates is compromised by the correlation coefficient, which represents the dependence among the variables.

In the portfolio selection problem, the estimation error caused by the uniform correlations 
does not compromise the quality of the optimal portfolio performance. 
When the correlation structure is non-uniform, we propose using the one factor model as an approximation for the correlation matrix, enabling the continued use of the computational method. 
To assess the model  performance in a more realistic context, we apply this approximation technique in a study using real financial data. 
Comparing the results to a benchmark, we observe that the portfolio formed by the rank constrained model demonstrates superior performance when the rankings are perfect.

\bibliographystyle{Chicago}

\bibliography{statqualref}
\end{document}